\documentclass{amsart}
\usepackage{amssymb,amsmath,amscd,amsthm}
\usepackage{amsfonts,epsfig,latexsym,graphicx,amssymb}%
\newtheorem{theo}{Theorem}[section]
\newtheorem{lemm}[theo]{Lemma}
\newtheorem{defi}[theo]{Definition}
\newtheorem{coro}[theo]{Corollary}
\newtheorem{prob}[theo]{Problem}
\newtheorem{rems}[theo]{Remarks}
\newtheorem{rema}[theo]{Remark}

\newcommand{\NN}{{\bf N}}

\newsavebox{\state}
\savebox{\state}(0,0){\put(0,0){\circle{15}}}
\newsavebox{\sstate}
\savebox{\sstate}(0,0){\put(0,0){\circle{8}}}
\newsavebox{\fstate}
\savebox{\fstate}(0,0){\put(0,0){\circle{15}}\put(0,0){\circle{12}}}
\newsavebox{\sfstate}
\savebox{\sfstate}(0,0){\put(0,0){\circle{8}}\put(0,0){\circle{11}}}
\begin{document}

\title{Finite automata with restricted two-way motion}

\author{David Damanik}

\address{Department of Mathematics, Rice University, Houston, TX~77005, USA}

\email{damanik@rice.edu}

\begin{abstract}
We consider finite two-way automata and measure the use of two-way motion by counting the number of left moves in accepting computations. Restriction of the automata according to this measure allows us to study in detail the use of two-way motion for the acceptance of regular languages in terms of state complexity. The two-way spectrum of a given regular language is introduced. This quantity reflects the change of size of minimal accepting devices if the
use of two-way motion is increased incrementally. We give examples for spectra, prove uniform upper and lower bounds and study their sharpness. We also have state complexity results for two-way automata with uniformly bounded use of two-way motion.
\end{abstract}

\maketitle

\section{Introduction}

The basic acceptance model for regular languages is the deterministic finite one-way automaton ($1DFA$). There are two natural ways of extending this model, allowing nondeterminism and two-way motion of the input head, respectively. Both of them do not allow the acceptance
of non-regular languages. However, the extensions ($1NFA$ and $2DFA$) can produce exponential savings in the number of states required to recognize a regular language.

State complexity issues have a long history going back at least to the pioneering work \cite{mf} of Meyer and Fischer. After equivalence of the several models of finite automata was established by Rabin-Scott \cite{rs}, Shepherdson \cite{sh} and Vardi \cite{v}, there was a considerable interest in results which demonstrate that a certain state complexity blow-up is indeed necessary in some cases. We mention the works of Sakoda-Sipser \cite{ss}, Sipser \cite{s}, and Micali \cite{m}.

On the other hand, in some cases not the whole power of the extension is needed, for some languages the minimal $1DFA$ is as small as the minimal $1NFA$ (resp., $2DFA$). Thus, for a
given regular language, one can ask how much nondeterminism (resp., two-way motion) is required to describe it adequately. Of course, one has to introduce well-motivated measures for nondeterminism (resp., two-way motion) first.

In their 1990 paper \cite{gkw}, Goldstine, Kintala and Wotschke handled the case of nondeterminism. The measure they used reflects the maximal (taken over the words in the language) number of branches in the minimal (regarding the number of branches) accepting computation. Starting with a $1DFA$ and allowing more and more nondeterminism according to
this measure they studied the decrease of state complexity in the nondeterminism spectrum. Concerning  two-way motion, we are going to follow similar lines.

By counting the number of left moves in the accepting computations we measure the two-way complexity of a given automaton. We get the class $2DFA(k)$ by restricting $2DFA$ according to this measure. An automaton from the class $2DFA(k)$ uses at most $k$ left moves in an accepting computation. The two-way spectrum of a regular language $L$ is then given by a monotonically decreasing infinite sequence where the entries are the sizes of minimal
$2DFA(k)$-descriptions of $L$. The rate of decrease reflects the change of size of a $2DFA$ describing $L$ when the two-way complexity is increased incrementally.

The paper is organized as follows. Section 2 introduces notation and recalls some classical results. Section 3 introduces our two-way complexity measure, the class $2DFA(k)$, and the spectrum $\sigma(L)$ of a regular language $L$, gives some examples of spectra with certain interesting properties, proves uniform upper and lower bounds for spectra and studies their sharpness. Section 4 compares the class $2DFA(k)$ with the classes $1DFA$, $1NFA$ and $2DFA$ under the viewpoint of possible state complexity savings (or, equivalently, blow-ups).

Due to the analogy mentioned above, some ideas and techniques are inspired by \cite{gkw}. Nevertheless, for the sake of self-containedness, full proofs of our theorems are provided.

\section{Finite two-way automata}

\subsection{Definitions}

\begin{defi}
A {\bf nondeterministic finite two-way automaton (2NFA)} is given by a quintuple
$M=(Q,\Sigma,\delta,q_0,F)$, where Q is the finite set of {\bf states}, $\Sigma$ is the finite {\bf input alphabet}, $\delta\;:\; Q \times \Sigma \; \rightarrow 2^{Q \times \{-1,+1\}}$ is the {\bf transition function}, $q_0$ is the {\bf starting state}, and F is the finite set of {\bf accepting states}. M is called {\bf de\-terministic (2DFA)}, if $|\delta (q,a)| \le 1$ for every $q \in Q$, $a\in \Sigma$. M is a {\bf one-way automaton (1NFA, resp., 1DFA)} if $\delta (q,a) \subseteq Q \times \{+1\}$ for every $q \in Q$, $a \in \Sigma$. The {\bf size} $|M|$ of a $2NFA$ $M$ is given by $|Q|$.
\end{defi}

\begin{defi}
A pair $(q,j) \in Q \times \NN$ is called {\bf configuration}, a finite sequence of configurations is called {\bf computation}. Let $M=(Q,\Sigma,\delta,q_0,F)$ be a $2NFA$ and $w=a_0 \ldots a_{n-1}$ an input. A computation $(p_0,j_0), \ldots , (p_m,j_m)$ is called {\bf computation of $M$ on $w$} if
$$
\begin{array}{ll} \bullet & p_0=q_0 \\ \bullet & j_0=0,\; j_m \le n\\
\bullet & \forall i \in \{0, \ldots ,m-1\}\; : \; 0 \le j_i < n \; and \\
& \exists \, (p,k) \in \delta (p_i,a_{j_i}) \; : \; p_{i+1}=p, \;
j_{i+1}=j_i + k
\end{array}
$$
If $j_m=n, \; p_m \in F$, the computation is called {\bf accepting}. We define

\begin{center}
$L(M) = \{w \in \Sigma ^*$ : {\rm there exists an accepting computation of $M$ on} $w\}.$
\end{center}
\end{defi}

\begin{rems}
\begin{enumerate}
\item We consider deterministic automata with partial functions as transition functions. Some authors {\rm (}e.g., {\rm \cite{hu69,hu79})} require total transition functions. In particular, there are no so-called TRAP-states in our automata.
\item A $2NFA$ does not possess endmarkers. Several authors {\rm (}e.g., {\rm \cite{ds,b1})} consider automata having $<w>$ on the input tape, where $w$ is the actual input. The transition function $\delta$ is then defined on $Q \times (\Sigma \cup \{<,>\})$, the symbols $<,>$ are part of the machine.
\item In order for an input to be accepted, the computation has to leave the input portion of the input tape off the rightmost symbol. In particular, a computation is not accepting if the computation loops on the input.
\end{enumerate}
\end{rems}

\subsection{Regularity}
If $A$ is a class of automata, one can define the class $L(A)$ of languages accepted by automata from $A$ by
$$
L(A) = \{L(M) \; | \; M \in A\}.
$$
Now, the class $Reg$ of regular languages is given by
$$
Reg = L(1DFA).
$$
However, one gets the same class of languages for every other class of automata introduced in the last subsection, that is,

\begin{equation}\label{equiv}
L(1DFA)=L(1NFA)=L(2DFA)=L(2NFA).
\end{equation}
By definition, the following inclusions hold:
$$
L(1DFA) \subseteq \begin{array}{c} L(1NFA) \\ L(2DFA) \end{array} \subseteq L(2NFA).
$$
Thus, in order to obtain (\ref{equiv}), for any given $2NFA$ $M_1$ one has to construct a $1DFA$ $M_2$ with $L(M_1)=L(M_2)$. This is possible; see Vardi \cite{v} (see also \cite{hu79}). However, there were earlier equivalence results which we list for completeness:

\begin{center}
\begin{tabular}{lll}
$\bullet$ & Rabin-Scott \cite{rs} & $2DFA \mapsto 1NFA$\\
$\bullet$ & Shepherdson \cite{sh} & $2DFA \mapsto 1DFA$\\
$\bullet$ & Vardi \cite{v} & $2NFA \mapsto 1DFA$
\end{tabular}
\end{center}
Rabin-Scott employed a crossing sequence analysis in order to eliminate two-way motion. This construction has the disadvantage that a former deterministic automaton is in general transformed into a nondeterministic automaton. The Shepherdson construction, on the other hand, does not introduce nondeterminism but is still not able to handle both two-way motion and nondeterminism. Vardi generalized the Shepherdson construction to automata from $2NFA$.

\subsection{State complexity: issues and concepts}
Given a regular language $L$ and two classes of computing devices $A_1,A_2$, where $A_1 \subseteq A_2$, one can ask the natural question whether the additional power of the class $A_2$ results in a description of the particular language $L$ with fewer states. Furthermore, the
maximum of this trade-off (taken over all regular languages which can be described by a finite subclass of automata from either $A_1$ or $A_2$) is an interesting object. In a slightly more general context, the definitions below provide upper and lower bounds for this quantity.

\begin{defi}[upper bound]
Let $A_1$, $A_2$ be two classes of automata, both accepting exactly the
class of regular languages. For monotonically increasing functions $f \; : \; \NN \rightarrow \NN$, $g \; : \;
\NN \rightarrow \NN$ we have
$$
\begin{array}{ccc}
A_1 & \rightarrow & A_2\\
f(n) && \le g(n)
\end{array},
$$
if and only if for every $n \in \NN$ the following holds: for every $M_1 \in A_1$
having $f(n)$ states, there exists $M_2 \in A_2$ obeying
$L(M_1)=L(M_2)$ and $|M_2| \le g(n)$.
\end{defi}

\begin{defi}[lower bound]
Let $A_1$, $A_2$ be two classes of automata, both accepting exactly the
class of regular languages. For monotonically increasing functions $f \; : \; \NN \rightarrow \NN$, $g \; : \;
\NN \rightarrow \NN$ we have
$$
\begin{array}{ccc}
A_1 & \rightarrow & A_2\\
f(n) && \ge g(n)
\end{array},
$$
if and only if for infinitely many $n \in \NN$ the following holds: there exists a regular language $L_n$ such that there is an automaton $M_1 \in A_1$ having $f(n)$ states and accepting $L_n$, and every automaton $M_2 \in A_2$ obeying $L(M_2)=L_n$ has at least $g(n)$ states.
\end{defi}
We will mostly consider the case $f= id$.

The constructions from the last subsection yield the following upper bounds:
$$
\begin{array}{lccll}
{\rm Rabin-Scott} & : & 2DFA & \rightarrow & 1NFA\\
&& n && \le n^{2n-1}\\
{\rm Shepherdson} & : & 2DFA & \rightarrow & 1DFA\\
&&n && \le (n+1)^{n+1}\\
{\rm Vardi}& : &2NFA & \rightarrow & 1DFA\\
&&n && \le (|\Sigma |+1) \cdot 2^{n^2+n}
\end{array}
$$
The constructions of Shepherdson and Vardi were improved by Birget in \cite{b1}, yielding the following upper bounds:
$$
\begin{array}{ccl}
2DFA & \rightarrow & 1DFA\\ n && \le n^n\\
2NFA & \rightarrow & 1DFA\\ n && \le 2^{(n-1)^2+n}
\end{array}
$$
There are several lower bound results which more or less demonstrate optimality of these constructions.

\begin{theo}[Meyer-Fischer '71]\label{mf}
$$
\begin{array}{ccc} 2DFA & \rightarrow & 1DFA\\ 5n+5 && \ge n^n \end{array}
$$
\end{theo}
This trade-off is obtained on the sequence
$$
L_n=\{0^{i_1}10^{i_2}1 \ldots 1 0^{i_n} 2^k 0^{i_k} \; : \; 1 \le k \le n, \; 1 \le i_j \le n\}
$$
and yields asymptotic optimality of the Shepherdson-construction. The next section shows an application of this bound and in particular motivates further study of exact sharpness of the upper bound.

\begin{theo}[Sakoda-Sipser '78]\label{Sakoda}
$$
\begin{array}{ccc} 2NFA & \rightarrow & 1DFA\\ n && \ge \frac{1}{2} 2^{(n-2)^2} \end{array}
$$
\end{theo}

\begin{defi}
A $2DFA$ is called {\bf Sweeping automaton (SA)}, if the reading direction is only changed at the leftmost or rightmost symbol of the input.
\end{defi}

\begin{theo}[Sipser '79]\label{Sipser}
$$
\begin{array}{ccc} 1NFA & \rightarrow & SA\\ n && \ge 2^n \end{array}
$$
\end{theo}

\begin{theo}[Micali '81]\label{Micali}
$$
\begin{array}{ccc} 2DFA & \rightarrow & SA\\ n && \ge c 2^n \end{array}
$$
\end{theo}
We want to add the following lower bound to this list. It will play a major role throughout the paper, in particular in the study of the sharpness of our lower bound for two-way spectra.

\begin{theo}\label{wcw}
There exists a constant $c > 0$ such that
$$
\begin{array}{ccc} 2DFA & \rightarrow & 1NFA\\ n && \ge 2^{cn} \end{array}
$$
\end{theo}
\begin{proof} Consider the following languages
$$
L_n = \{awbwa \; : \; w \in \{0,1\}^*, \, |w|=n\}.
$$
It is known that every $1NFA$ accepting $L_n$ has at least $2^n$ states \cite{w}. The $2DFA$ $M$ given in Figure~\ref{big} has $\mathcal{O}(n)$ states and accepts $L_n$.
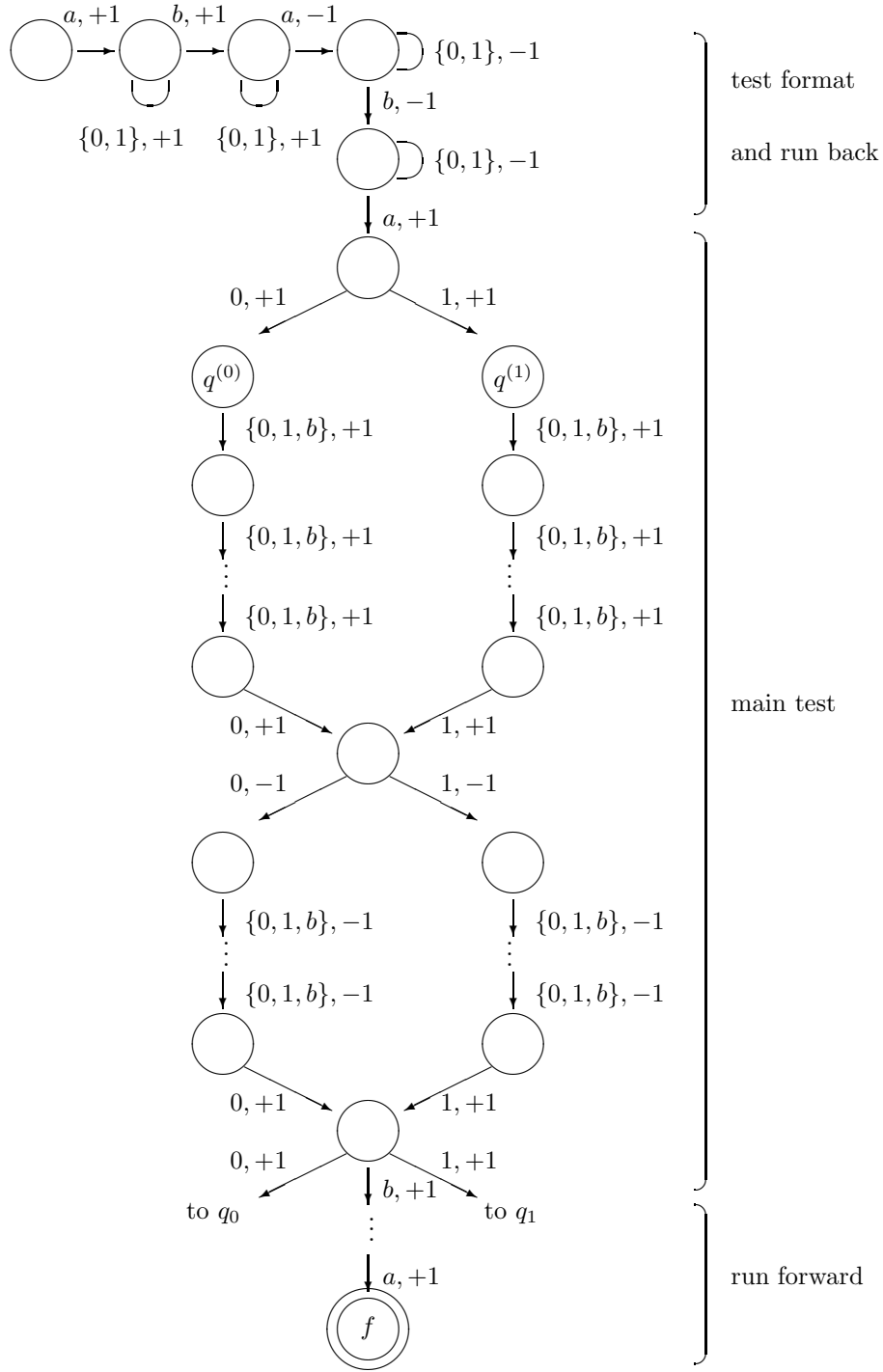
\begin{figure}
\begin{center}
\begin{picture}(150,200)
%
%
\put(15,190){\usebox{\sstate}}
\put(30,190){\usebox{\sstate}}
\put(45,190){\usebox{\sstate}}
\put(60,190){\usebox{\sstate}}
\put(60,175){\usebox{\sstate}}
\put(60,160){\usebox{\sstate}}
\put(40,145){\usebox{\sstate}\makebox(0,4){$q^{(0)}$}}
\put(80,145){\usebox{\sstate}\makebox(0,4){$q^{(1)}$}}
\put(40,130){\usebox{\sstate}}
\put(80,130){\usebox{\sstate}}
\put(40,105){\usebox{\sstate}}
\put(80,105){\usebox{\sstate}}
\put(60,93){\usebox{\sstate}}
\put(40,78){\usebox{\sstate}}
\put(80,78){\usebox{\sstate}}
\put(40,53){\usebox{\sstate}}
\put(80,53){\usebox{\sstate}}
\put(60,41){\usebox{\sstate}}
\put(60,13){\usebox{\sfstate}\makebox(0,6){$f$}}
%
%
\put(20,192){\vector(1,0){5}}
\put(35,192){\vector(1,0){5}}
\put(50,192){\vector(1,0){5}}
%
%
\put(60,187){\vector(0,-1){5}}
\put(60,172){\vector(0,-1){5}}
\put(40,142){\vector(0,-1){5}}
\put(80,142){\vector(0,-1){5}}
\put(40,127){\vector(0,-1){5}}
\put(80,127){\vector(0,-1){5}}
\put(40,117){\vector(0,-1){5}}
\put(80,117){\vector(0,-1){5}}
\put(40,75){\vector(0,-1){5}}
\put(80,75){\vector(0,-1){5}}
\put(40,65){\vector(0,-1){5}}
\put(80,65){\vector(0,-1){5}}
\put(60,38){\vector(0,-1){5}}
\put(60,26){\vector(0,-1){5}}
%
%
\put(57,159){\vector(-2,-1){12}}
\put(63,159){\vector(2,-1){12}}
\put(43,104){\vector(2,-1){12}}
\put(77,104){\vector(-2,-1){12}}
\put(57,92){\vector(-2,-1){12}}
\put(63,92){\vector(2,-1){12}}
\put(43,52){\vector(2,-1){12}}
\put(77,52){\vector(-2,-1){12}}
\put(57,40){\vector(-2,-1){12}}
\put(63,40){\vector(2,-1){12}}
%
%
\put(30,188){\oval(5,7)[b]}
\put(45,188){\oval(5,7)[b]}
\put(64,192){\oval(7,5)[r]}
\put(64,177){\oval(7,5)[r]}
\put(105,182){\oval(3,25)[r]}
\put(105,101){\oval(3,132)[r]}
\put(105,22){\oval(3,22)[r]}
%
%
\put(20,179){$\{0,1\},+1$}
\put(39,179){$\{0,1\},+1$}
\put(43,139){$\{0,1,b\},+1$}
\put(83,139){$\{0,1,b\},+1$}
\put(43,124){$\{0,1,b\},+1$}
\put(83,124){$\{0,1,b\},+1$}
\put(43,113){$\{0,1,b\},+1$}
\put(83,113){$\{0,1,b\},+1$}
\put(43,71){$\{0,1,b\},-1$}
\put(83,71){$\{0,1,b\},-1$}
\put(43,61){$\{0,1,b\},-1$}
\put(83,61){$\{0,1,b\},-1$}
\put(18,196){$a,+1$}
\put(33,196){$b,+1$}
\put(48,196){$a,-1$}
\put(69,191){$\{0,1\},-1$}
\put(62,184){$b,-1$}
\put(69,176){$\{0,1\},-1$}
\put(62,168){$a,+1$}
\put(40,118){$\vdots$}
\put(79,118){$\vdots$}
\put(40,66){$\vdots$}
\put(79,66){$\vdots$}
\put(35,31){to $q_0$}
\put(76,31){to $q_1$}
\put(110,187){test format}
\put(110,177){and run back}
\put(110,101){main test}
\put(110,22){run forward}
\put(41,157){$0,+1$}
\put(70,157){$1,+1$}
\put(41,98){$0,+1$}
\put(70,98){$1,+1$}
\put(41,90){$0,-1$}
\put(70,90){$1,-1$}
\put(41,46){$0,+1$}
\put(70,46){$1,+1$}
\put(41,38){$0,+1$}
\put(70,38){$1,+1$}
\put(62,34){$b,+1$}
\put(60,28){$\vdots$}
\put(62,22){$a,+1$}
\end{picture}
\caption{Automaton $M$ from the proof of Theorem~\ref{wcw}}\label{big}
\end{center}
\end{figure}
%
%
%
%
%
%
%
%
%
%
%
%
Moreover, it is quite easy to see that an accepting computation has $\mathcal{O}(n^2)$ left moves.
\end{proof}

\section{The two-way spectrum of a regular language $L$}

\subsection{Definition and computability}

\begin{defi}
Let $M=(Q,\Sigma,\delta,q_0,F)$ be a 2DFA and $w=a_0 \ldots a_{n-1} \in \Sigma^*$. Let $(q_0,i_0),\ldots,(q_m,i_m)$ be the corresponding sequence of configurations in M. Define
$\lambda (w) = \# \{j:0\le j\le m,\; i_j=-1\}$ and $\lambda(M) = \sup \{\lambda(w):w\in L(M)\}$.
\end{defi}

\begin{defi}
Define for $k \in \NN \cup \{ \infty \}$,
$$
2DFA(k) = \{M : M \; {\rm is} \; {\rm a} \; 2DFA \; {\rm with} \; \lambda(M) \le k\}.
$$
\end{defi}

\begin{defi}
The {\bf two-way spectrum} of $L$ is defined by
$$
\sigma (L)= (\sigma_0(L),\sigma_1(L), \ldots ;\sigma_\infty(L)),
$$
where
$$
\sigma_k(L) = \min \{|M|: M \in 2DFA(k), \; L(M)=L\}.
$$
\end{defi}

\begin{rema}
The sequence $\sigma_k(L)$ is monotonically decreasing since $2DFA(k_1) \subseteq 2DFA(k_2)$ if $k_1 \le k_2$.
\end{rema}

\begin{lemm}\label{34}
Let M be a 2DFA. $\lambda(M) < \infty$ is decidable.
\end{lemm}

\begin{proof}
Recall the concept of a crossing sequence (see \cite{hu79}): Write, for a computation of $M$ on some word $w$, the list of states $M$ is currently in during the computation at hand below each boundary between two consecutive input symbols. This list is called a crossing sequence. It is clear that the first time a boundary is crossed, the head must be moving right. Subsequent crossings must be in alternate directions. If the input $w$ is accepted, it follows that all crossing sequences below the input are of odd length and that no two odd- and no two even-numbered elements in one of these crossing sequences are indentical. Call crossing sequences with these two properties \textit{valid}. Thus, if $n$ is the number of states of $M$, the length of a valid crossing sequence is bounded by an effective constant $B = B(n)$ and the number of valid crossing sequences is bounded by an effective constant $C = C(n)$.

We claim that there is an effective constant $D = D(n)$ such that $\lambda(M) = \infty$ if and only if there is a word $w \in L(M)$ having length bounded by $D$ such that the accepting computation of $M$ on $w$ has a repeated crossing sequence which contains at least three elements. This condition is decidable.

Consider first an accepting computation of $M$ on a word $w = w_1 \ldots w_n$ such that the crossing sequences right of $w_i$ and right of $w_k$ coincide and contain at least three elements. Then we can consider the words $w^{(l)} = w_1 \ldots w_i (w_{i+1} \ldots w_k)^l w_{k+1} \ldots w_n$ which all belong to $L(M)$. For every $N \in \NN$, there is $l$ such that $\lambda(w^{(l)}) > N$. Hence, $\lambda(M) = \infty$.

Conversely, let $\lambda(M) = \infty$. Since the length of valid crossing sequences is uniformly bounded by $B(n)$, we can find a word $w \in L(M)$ such that the accepting computation of $M$ on $w$ has a repeated crossing sequence which has at least three elements. While the argument above used a pumping argument, we shall now employ an inverse-pumping argument to prune the word $w$ in order to obtain a word $w' \in L(M)$ which still has a repeated crossing sequence with at least three elements and whose length we can bound by an effective constant depending only on $n$. First, one can delete parts from $w$ until exactly one repeated crossing sequence is left. Note that this yields an accepting computation. Next, one deletes parts from the stretches where there are consecutive crossing sequences of length one so that there are no repeated states left on such stretches. This will again yield an accepting computation. The resulting word $w' \in L(M)$ will have its length bounded by $n (2+C(n))$.

\end{proof}

\begin{lemm}\label{35}
Let M be a 2DFA obeying $\lambda(M) < \infty$ and let $k \in \NN$. $\lambda(M)=k$ is decidable.
\end{lemm}

\begin{proof} Let $T_k(M) = \{w \; : \;w \in L(M),\, \lambda(w) \le k\}$. Obviously,
$$
\lambda(M)=k \; \Leftrightarrow \; T_{k-1}(M) \not= T_k(M) = L(M).
$$
We are therefore done with the proof if we show that $T_k(M)$ is regular. The idea is the following: Count the left moves in a homomorphic image and restrict their number by $k$. We introduce the following alphabet:
$$
\Sigma_T = \{[p,a,r,q] \; : \; \delta(p,a)=(q,r)\}.
$$
Let $R\subseteq \Sigma_T^*$ be the (regular) set of accepting computations of $M$:
$$
R = \{[p_0,a_0,r_0,p_1][p_1,a_1,r_1,p_2] \cdots [p_{n-1},a_{n-1},r_{n-1},p_n] \; :
$$
$$
\delta(p_i,a_i)=(p_{i+1},r_i), \; p_0=q_0, \; p_n \in F, \; r_0=r_{n-1}=+1\} \; \cup \{\varepsilon \; : \; q_0 \in F\}
$$
Let $S_k=\{1^i \; : \; 0 \le i \le k\}$ and
$$
h:\Sigma_T^* \rightarrow \{1\}^* , \;\; [p,a,r,q] \mapsto \left\{ \begin{array}{cc}
1, & r=-1\\ \varepsilon, & r=+1 \end{array} \right. .
$$
With the finite transducer $E_k$ from Figure~\ref{transek},
%
%
%
%
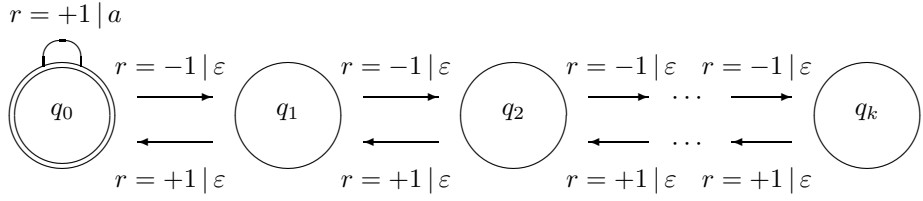
\begin{figure}
\begin{center}
  \begin{picture}(128,40)
%
%
\put(40,17){\usebox{\state}\makebox(0,8){$q_1$}}
\put(70,17){\usebox{\state}\makebox(0,8){$q_2$}}
\put(117,17){\usebox{\state}\makebox(0,8){$q_k$}}
\put(10,17){\usebox{\fstate}\makebox(0,8){$q_0$}}
%
%
\put(20,23){\vector(1,0){10}}
\put(50,23){\vector(1,0){10}}
\put(80,23){\vector(1,0){8}}
\put(99,23){\vector(1,0){8}}
\put(30,17){\vector(-1,0){10}}
\put(60,17){\vector(-1,0){10}}
\put(88,17){\vector(-1,0){8}}
\put(107,17){\vector(-1,0){8}}
%
%
%
\put(3,33){$r=+1 \, | \, a$}
\put(17,26){$r=-1 \, | \, \varepsilon$}
\put(47,26){$r=-1 \, | \, \varepsilon$}
\put(77,26){$r=-1 \, | \, \varepsilon$}
\put(95,26){$r=-1 \, | \, \varepsilon$}
\put(17,11){$r=+1 \, | \, \varepsilon$}
\put(47,11){$r=+1 \, | \, \varepsilon$}
\put(77,11){$r=+1 \, | \, \varepsilon$}
\put(95,11){$r=+1 \, | \, \varepsilon$}
\put(91,16){$\cdots$}
\put(91,22){$\cdots$}
%
%
%
\put(10,27){\oval(5,7)[t]}
  \end{picture}
\caption{Transducer $E_k$}\label{transek}
\end{center}
\end{figure}
we have $T_k(M)=E_k(R \cap h^{-1}(S_k))$. Thus, $T_k(M)$ is regular.
\end{proof}

\begin{theo}
$\sigma(L)$ is computable.
\end{theo}
\begin{proof} Let  $M$ be the minimal $1DFA$ with $L(M)=L$. Then, $\sigma_0(L)=|M|$. Let $M_1,\ldots,M_n$ be the (finitely many and effectively determinable) $2DFA$ obeying $|M_i| \le \sigma_0(L)$ and $L(M_i)=L$. Compute the complexities $\lambda(M_i)$, using Lemma~\ref{34} and Lemma~\ref{35}. Then, determine minimal pairs $(|M_i|,\lambda(M_i))$, where minimality is understood according to the order $(x_1,y_1) \le (x_2,y_2) \; :\Leftrightarrow \; x_1 \le x_2, \, y_1 \le y_2$. This gives the sequence
$$
(s_1,\lambda_1),\ldots,(s_j,\lambda_j)
$$
of minimal pairs obeying $s_1 > \ldots > s_j$ and $0=\lambda_1 < \ldots < \lambda_j$. Let $\lambda_{j+1} = \infty$. Then, $\sigma_k(L) = s_l \; \forall \lambda_l \le k < \lambda_{l+1}$ and $\sigma_\infty(L) = s_j$.
\end{proof}

\subsection{Examples}

\subsubsection{The constant spectrum}
The sequence of languages $L_n=\{1^{n-1}\} \subseteq 1^*$, $n \in \NN$, exhibits the occurrence of a constant spectrum $(n,n,\ldots;n)$. One checks (cf.~\cite{b2}):

\begin{itemize}
\item $\sigma_0(L_n) \le n,$
\item $\sigma_\infty(L_n) \ge n.$
\end{itemize}

\subsubsection{The collapse at $n \in \NN$}
We aim at the construction of a spectrum which is constant for $\{k \; | \; k=0,\ldots,n-1\}$ and for $\{k \; | \; k=n,n+1,\ldots, \infty \}$, and which has a major jump in the transition from  $n-1$ to $n$. Consider the languages $L_n=\{0,1\}^*1\{0,1\}^{n-1}\$ $, $n \in \NN$. The sequence of lemmas below establishes the desired form of the spectrum.

\begin{lemm}
$\sigma_0(L_n) \le 2^n+1.$
\end{lemm}

\begin{proof} Define the $1DFA$ $M=(Q,\Sigma,\delta,q_0,F)$ by $Q = \{0,1\}^n \cup \{f\}$, $\Sigma = \{0,1,\$\}$, $q_0 = (0,\ldots,0)$, $F = \{f\}$ and
$$
\delta((a_1,\ldots,a_n),a) = \left\{ \begin{array}{ll} (a_2, \ldots,a_n,a) & a \in \{0,1\}\\ f & {\rm if} \; a=\$ \; {\rm and} \; a_1=1\\ {\rm undefined} & {\rm else} \end{array}\right.
$$
Obviously, $L(M)=L_n$ and $|M|=2^n+1$.
\end{proof}

\begin{lemm}
$\sigma_{n-1}(L_n) \ge 2^n+1$.
\end{lemm}

\begin{proof} Let $M=(Q,\{0,1,\$\},\delta,q_0,F) \in 2DFA(n-1)$ obey $L(M)=L_n$. Then, by $\lambda(M) \le n-1$,
$$
\exists \, \hat{w} \; : \; \lambda(\hat{w})=\max \{ \lambda(v) \; : \; v \in
\{0,1\}^*\}.
$$
$M$ eventually leaves $\hat{w}$ off its rightmost symbol and, by definition of $\hat{w}$, $M$ then moves to the right until $\$$ is read. Now, every $w \in \{0,1\}^*$ is prefix of a word in $L_n$. Thus, the computation of $M$ on $w$ reaches the rightmost symbol of $w$ in a state $q_w$. We want to show that $\{ \hat{w} v \; | \; v \in \{0,1\}^n\}$ is a set of words that distinguishes between states. Suppose to the contrary
$$
q_{\hat{w} v_1}=q_{\hat{w} v_2} \; {\rm for} \; v_1,v_2 \in \{0,1\}^n, \; v_1
\not= v_2.
$$
Without loss of generality, $v_1=x_11y$, $v_2=x_2 0y$ with $|x_1| =|x_2|$. $\hat{w}v_1 0^{n-|y|-1}\$$ and  $\hat{w}v_2 0^{n-|y|-1}\$$ are therefore either both accepted or both rejected, contradiction. Hence, $|Q| \ge 2^n$. Since $L_n \not= \emptyset$ and all the $q_{\hat{w} v}$
are non-accepting, we conclude $|Q| \ge 2^n +1$.
\end{proof}

\begin{lemm}\label{n+2}
$\sigma_n(L_n) \le n+2$.
\end{lemm}

\begin{proof} Consider the automaton in Figure~\ref{3}.
%
%
%
%
%
\begin{figure}
\begin{center}
  \begin{picture}(80,67)
%
%
\put(10,40){\usebox{\state}\makebox(0,8){$q_0$}}
\put(10,10){\usebox{\fstate}\makebox(0,8){$q_1$}}
\put(70,10){\usebox{\state}\makebox(0,8){$q_n$}}
\put(70,40){\usebox{\state}\makebox(0,8){$q_{n+1}$}}
%
%
\put(18,13){\vector(1,0){11}}
\put(50,13){\vector(1,0){11}}
\put(60,40){\vector(-2,-1){35}}
%
%
\put(10,34){\vector(0,-1){11}}
\put(70,23){\vector(0,1){11}}
%
%
\put(10,50){\oval(5,7)[t]}
\put(70,50){\oval(5,7)[t]}
%
%
\put(3,57){$\{0,1\},+1$}
\put(12,28){$\$,-1$}
\put(19,16){$\{0,1\},-1$}
\put(46,16){$\{0,1\},-1$}
\put(73,28){$1,+1$}
\put(63,57){$\{0,1\},+1$}
\put(50,40){$\$,+1$}
\put(37,12){$\cdots$}
\end{picture}
\end{center}
\caption{Automaton from the proof of Lemma~\ref{n+2}}\label{3}
\end{figure}
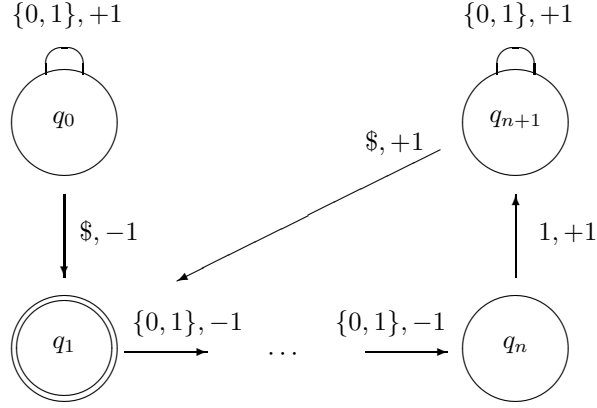
One easily checks that this automaton accepts $L_n$ and belongs to $2DFA(n)$.
\end{proof}

\begin{lemm}
$\sigma_\infty(L_n) \ge n+2$.
\end{lemm}

\begin{proof} When a $2DFA$ $M$ reaches $'\$'$, $M$ has to move $n$ symbols backwards, since otherwise $|M|\ge 2^n +1$ (compare the proof of Lemma 3.9). In order to count $n$ arbitrary symbols, a $2DFA$ needs $n+1$ states \cite{b2}. In those states (with the exception of the last
$(1,+1)$-transition) only left moves are defined for inputs $0,1$. Thus, $M$ has at least one more state, since $L_n$ is not empty.
\end{proof}

\subsubsection{The collapse at $\infty$}
Consider the Meyer-Fischer languanges
$$
L_n=\{0^{i_1}1 0^{i_2}1 \ldots 1 0^{i_n} 2^m 0^{i_m} \; | 1 \le m \le n, \; 1 \le i_j \le n\},
$$
which yield asymptotic optimality of the Shepherdson-construction. In the next subsection it will be shown that for $L_n' = (\$L_n\$)^*$ the two-way complexity collapses at infinity:

\begin{itemize}
\item $\sigma_k(L_n')= \Omega (n^n) \mbox{ for every } k \in \NN$,
\item $\sigma_\infty(L_n') = \mathcal{O}(n)$.
\end{itemize}

\subsubsection{The decreasing spectrum}
The languages
$$
R_n=c\{awbwa \; : \; w \in \{0,1\}^1\} c \ldots c \{awbwa \; : \; w \in \{0,1\}^n\}c
$$
have spectra with many points of decrease. The idea is the following: The more two-way motion is allowed, the more blocks can be tested in two-way fashion. This argument will be formalized at the end of this section.

\subsection{Upper bounds for $\sigma(L)$}
The Birget improvement of the Shepherdson construction immediately yields

\begin{theo}\label{ubfs}
Let $L \subseteq \Sigma^*$ be regular with $n = \sigma_\infty (L)$. Then,
$$
\sigma(L) \le (n^n, n^n, \ldots; n).
$$
\end{theo}

\begin{rema}
Fixing the spectrum at $\sigma_0(L)=n$ yields the trivial upper bound $(n,n,\ldots;n)$, which is sharp by Example 3.2.1.
\end{rema}

In order to study sharpness of the bound from Theorem \ref{ubfs}, we have to look for languages $L$ for which $\sigma(L)$ is constant with the exception of $\sigma_\infty (L)$. The following lemma gives a class of such languages.

\begin{lemm}\label{312}
Let $L \subseteq \Sigma^*$ be regular and $\$ \not\in \Sigma$. Then, there exist $k,n \in \NN$ with
$$
\sigma((\$L\$)^*)=(k,k,\ldots;n).
$$
\end{lemm}

\begin{proof} Let $L \subseteq \Sigma^*$ be an arbitrary regular language. Define $L' = (\$L\$)^*$. We will freely use the following simple properties

\begin{eqnarray*}
v,w \in L' & \Rightarrow & vw \in L',\\
vw \in L', \; v \in L' & \Rightarrow & w \in L',\\
vw \in L', \; w \in L' & \Rightarrow & v \in L'.
\end{eqnarray*}
It suffices to show
$$
\sigma_k(L') \ge \sigma_0(L') \;\; \forall k \in \NN.
$$
To a given $M=(Q,\Sigma \cup \{\$\},\delta,q_0,F) \in 2DFA(k)$ with $L(M)=L'$ we will construct an equivalent $2DFA(0)$ $N=(Q,\Sigma \cup \{\$\}, \delta', q_0',F')$ with $L(N)=L'$. In particular, we have $|N|=|M|$. Define the sets $S$ and $T$ as follows. The set $S$ consists of the states $q \in Q$ with the following property: There exists a word $w_q \in L'$ such that $M$ runs on $w_q$ from $q_0$ to $q$:
%
%
%
%
%
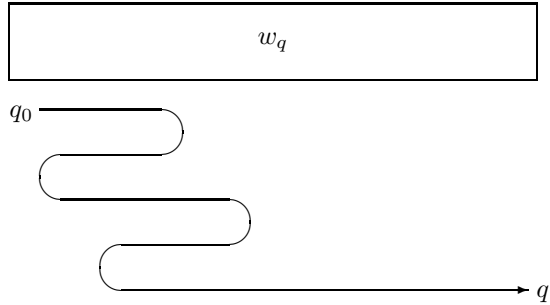
\begin{center}
  \begin{picture}(80,45)
%
%
\put(5,30){\framebox(70,10){$w_q$}}
\put(5,25){$q_0$}
\put(75,1){$q$}
%
%
\put(9,26){\line(1,0){16}}
\put(12,20){\line(1,0){13}}
\put(12,14){\line(1,0){22}}
\put(20,8){\line(1,0){14}}
\put(20,2){\vector(1,0){54}}
\put(25,23){\oval(6,6)[r]}
\put(12,17){\oval(6,6)[l]}
\put(34,11){\oval(6,6)[r]}
\put(20,5){\oval(6,6)[l]}
  \end{picture}
\end{center}
The set $T$ consists of the states $q \in Q$ with the following property: There exist $w_q \in L'$ and $f_q \in F$ such that:
%
%
%
%
%
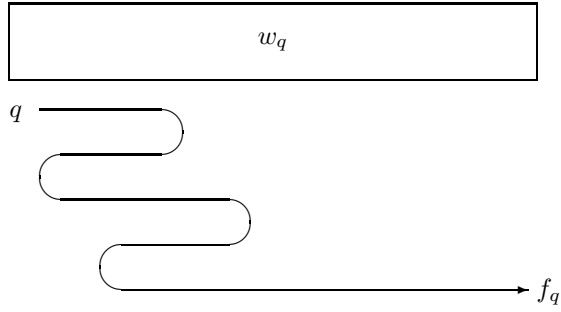
\begin{center}
  \begin{picture}(80,45)
%
%
\put(5,30){\framebox(70,10){$w_q$}}
\put(5,25){$q$}
\put(75,1){$f_q$}
%
%
\put(9,26){\line(1,0){16}}
\put(12,20){\line(1,0){13}}
\put(12,14){\line(1,0){22}}
\put(20,8){\line(1,0){14}}
\put(20,2){\vector(1,0){54}}
\put(25,23){\oval(6,6)[r]}
\put(12,17){\oval(6,6)[l]}
\put(34,11){\oval(6,6)[r]}
\put(20,5){\oval(6,6)[l]}
  \end{picture}
\end{center}
The transition function $\delta'$ results from $\delta$ by omitting left moves. Thus $\delta'$ is the transition function of a one-way automaton and it makes sense to write $\delta'(P,W)$ for some subset $P$ of $Q$ and a set of words $W$. Below we will show

{\it Claim 1.} $L'=\{w \; | \; \delta'(S,w) \cap T \not= \emptyset\}.$

{\it Claim 2.} Let $S_0$ be a minimal subset of $S$ which satisfies Claim 1. Let $s_0 \in S_0$ be arbitrary. Then, $L'=\{w\; | \; \delta'(s_0,w) \cap T \not= \emptyset\}.$

Choosing $q_0'=s_0 \in S_0$ arbitrarily and defining $F' = T$, we get the desired automaton $N \in 2DFA(0)$.

{\it Proof of Claim 1.} '$\subseteq$' Let $v \in L'$. Then, $v^{k+1} \in L'$. Since $\lambda(M) \le k$ there is $j$ such that there is no left move on the $j$-th $v$-block. Let $q_1$ (resp., $q_2$) be the state $M$ is in as the computation enters (resp., exists) this $v$-block. Then, $q_1 \in S$, $q_2 \in T$, and $\delta'(q_1,v)=q_2$. Thus, $v \in  \{w \; | \; \delta'(S,w) \cap T \not= \emptyset\}$.

'$\supseteq$' Let $v \in \{w \; | \; \delta'(S,w) \cap T \not= \emptyset\}$. By definition of $S$ and $T$, there exist $x,y \in L'$ such that $xwy \in L'$. Hence, $w \in L'$.

{\it Proof of Claim 2.} By Claim 1 we only have to show that $w \in L'$ implies $\delta'(s_0,w) \cap T \not= \emptyset$. Suppose there is $w \in L'$ with $\delta'(s_0,w) \cap T = \emptyset$. Then,

\begin{equation}\label{empty}
\delta'(s_0,wL') \cap T = \emptyset.
\end{equation}
Since $S_0$ satisfies Claim 1, we have

\begin{equation}\label{empty2}
\delta'(S_0,wv) \cap T \not= \emptyset \;\; \forall v \in L'.
\end{equation}
(\ref{empty}) and (\ref{empty2}) imply $\delta'(S_0 \backslash \{s_0\}, wv) \cap T \not= \emptyset \;\; \forall v \in L'$. Let $\hat{S_0} = \delta'(S_0 \backslash \{s_0\},w)$. Then,
$\hat{S_0} \subseteq \delta'(S,L') \subseteq S$ and $\delta'(\hat{S_0},v) \cap T \not= \emptyset \;\; \forall v \in L'$, and we get the chain of inclusions
$$
L' \subseteq \{v \; | \; \delta'(\hat{S_0},v) \cap T \not= \emptyset\} \subseteq \{v \; | \; \delta'(S,v) \cap T \not= \emptyset\} =L'.
$$
Thus, equality holds. In particular, $\hat{S_0}$ satisfies Claim 1. By $\# \hat{S_0} \le \# (S_0 \backslash \{s_0\}) < \# S_0$, this contradicts the  minimality of $S_0$.
\end{proof}

\begin{lemm}\label{313}
Let $L \subseteq \Sigma^*$ be regular and $\sigma_0(L) \ge f(\sigma_\infty(L))$ with a monotonically increasing function $f$. Let $L' = (\$L \$) ^* $. Then,
$$
\sigma_0(L') \ge f(\sigma_\infty(L')-1)+1.
$$
\end{lemm}

\begin{proof} By adding one state, the minimal automata for $L$ can be modified to yield automata for $L'$. For $k=0$, one-way motion is preserved. We therefore have $\sigma_0(L') \le \sigma_0(L) +1$ and $\sigma_\infty (L') \le \sigma_\infty (L) +1$.
We will now show

\begin{equation}\label{inf}
\sigma_0(L') \ge \sigma_0(L) +1,
\end{equation}
yielding
$$
\sigma_0(L') = \sigma_0(L) + 1 \ge f(\sigma_\infty(L)) +1 \ge f(\sigma_\infty (L')-1)+1.
$$

{\it Proof of (\ref{inf}):} We will use the fact that for regular $L \subseteq \Sigma^*$, $\sigma_0(L)$ is equal to the number of non-empty left quotients $w \backslash L, \; w \in \Sigma^*$.

For every $w$ not containing $\$$, we have $(\$ w) \backslash (\$L \$) ^* = w \backslash L \$ (\$L \$) ^*$. Thus, $L'$ has at least as many non-empty left quotients as $L$. Since $L'= \varepsilon \backslash L'$ is another one, $L'$ has at least one more.
\end{proof}

\begin{theo}
Let $L_n$ be regular and $f$ monotonically increasing with $\sigma_\infty(L_n) = n$ and $\sigma_0(L_n) \ge f(n)$. Then, there exists a sequence of regular languages $L_n'$ with
$$
\sigma(L_n') \ge (f(n-1)+1, f(n-1)+1, \ldots;n)
$$
for infinitely many $n$.
\end{theo}

\begin{proof} The assertion follows from Lemmas \ref{312} and \ref{313}.
\end{proof}

\begin{rema}
We have reduced sharpness issues for the upper bound for $\sigma(L)$ given in Theorem~\ref{ubfs} to the problem
$$
\begin{array}{ccc} 2DFA & \rightarrow & 1DFA\\
n && \ge \; ? \end{array}
$$
that is, the question to what extent the Shepherdson construction is sharp.
\end{rema}

\subsection{Lower bounds for $\sigma(L)$}
We now aim at proving a similar uniform lower bound for two-way spectra. We first focus on the following upper bound. By reversing the view we will deduce the desired lower bound below.
$$
\begin{array}{ccc} 2DFA(k) & \rightarrow & 1DFA\\
n && \le \; ? \end{array}.
$$
Let $M$ be a $2DFA(k)$ over $\Sigma$. Intuitively, we can construct an equivalent $1DFA$ $M'$ in two steps as follows:

{\it Step 1:} Store the last $k$ symbols in the finite control. Simulate left-moves in $M$ by using $\varepsilon$-moves.

{\it Step 2:} Eliminate $\varepsilon$-moves.

\begin{theo}\label{316}
For $k \in \NN$ we have:
$$
\begin{array}{ccl} 2DFA(k) & \rightarrow & 1DFA\\ n && \le n\cdot (k+1)\cdot (|\Sigma |+1)^{k+1}
\end{array}
$$
\end{theo}

\begin{rems}

\begin{enumerate}
\item The trade-off is linear if $k$ and $|\Sigma|$ are fixed. The analogue in the nondeterminism case does not hold \cite{gkw}!
\item The trade-off is exponential in $k$.
\item Compared with the Shepherdson-construction, our method has no advantage for $k \ge \Omega (n \cdot \ln n)$.
\end{enumerate}

\end{rems}

\begin{proof} We formalize the two steps from above. Let $M=(Q,\Sigma,\delta, q_0,F) \in 2DFA(k)$.

{\it Step 1 : Construction of an equivalent $1DFA$ $M'=(Q',\Sigma,\delta', q_0',F')$ with $|Q'| = |Q| \cdot (k+1) \cdot (| \Sigma |+1)^{k+1}$}

Let $ \$ \not\in \Sigma$. Define $Q' = Q \times \{0,\ldots ,k\} \times (\Sigma \cup \{\$\})^{k+1}$, $q_0' = (q_0,0,\$, \ldots,\$)$, $F' = \{ (f,0,a_k,\ldots,a_1,\$) \; : \; f \in F; \; a_i \in \Sigma \cup \{\$\}, \; 1 \le i \le k\}$ and

\begin{small}
$$
\begin{array}{lll} \delta'((p,j,a_k,\ldots,a_0),\varepsilon) & = & \left\{
\begin{array}{cc} (q,j-1,a_k,\ldots,a_0) & {\rm if} \; \delta(p,a_j)=(q,+1), \; j>0\\
(q,j+1,a_k,\ldots,a_0) & {\rm if} \; \delta(p,a_j)=(q,-1), \; j>0 \end{array}
\right.\\
\delta'((p,0,a_k,\ldots,a_1,a_0),\varepsilon) & = & \left\{
\begin{array}{cc}
(q,0,a_{k-1},\ldots,a_0,\$) & {\rm if} \; \delta(p,a_0)=(q,+1)\\
(q,1,a_k,\ldots,a_0) & {\rm if} \; \delta(p,a_0)=(q,-1)
\end{array}
\right.\\
\delta'((p,0,a_k,\ldots,a_1,\$),a) & = & \left\{
\begin{array}{cc}
(q,0,a_{k-1},\ldots,a_1,a,\$) & {\rm if} \; \delta(p,a)=(q,+1)\\
(q,1,a_k,\ldots,a_1,a) & {\rm if} \; \delta(p,a)=(q,-1)
\end{array}
\right.
\end{array}
$$
\end{small}
Obviously, $M'$ is deterministic and satisfies $L(M')=L(M)$.

{\it Step 2 : Elimination of $\varepsilon$-moves}

Define $M''=(Q',\Sigma,\delta'',q_0',F'')\in 1DFA$, where
$$
\begin{array}{lll} \varepsilon(q) & = & \{q' \; : \; \delta'(q,\varepsilon^*)=q' \;
{\rm and} \; \delta'(q',\varepsilon) \; {\rm is} \; {\rm undefined} \} \\
\delta''(q,a) & = & \left\{
\begin{array}{cc} \delta'(q,a) & {\rm if} \; \delta'(q,a) \not= \emptyset\\
\delta'(\varepsilon(q),a) &  {\rm if} \; \varepsilon(q) \not= \emptyset
\end{array}
\right. \\
F'' & = & F' \cup \{q_0' \; : \; \varepsilon(q_0') \cap F' \not=
\emptyset\}
\end{array}
$$
$\delta''$ is well-defined and $M''$ has $|Q| \cdot (k+1)\cdot (|\Sigma |+1)^{k+1}$ states.
\end{proof}

\begin{coro}
For $k \in \NN$ we have
$$
\begin{array}{ccl} 1DFA & \rightarrow & 2DFA(k)\\
n && \ge n\cdot (k+1)^{-1} \cdot (|\Sigma |+1)^{-(k+1)}
\end{array}
$$
\end{coro}
This yields the following lower bound for $\sigma(L)$:

\begin{theo}\label{usfds}
Let $L \subseteq \Sigma^*$ be regular with $\sigma_0(L) \ge n^n$. Then,
$$
\sigma(L) \ge (n^n, \frac{1}{2} c^2 \cdot n^n, \frac{1}{3}c^3 \cdot n^n,
\frac{1}{4}c^4 \cdot n^n,\ldots,n,n,\ldots;n),
$$
where $c=(|\Sigma|+1)^{-1}$.
\end{theo}

\begin{rema}
Fixing the other end of the spectrum again yields the trivial and sharp lower bound $(n,n,\ldots;n)$.
\end{rema}
We are going to study the sharpness of the bound from Theorem \ref{usfds}. To this end, we will follow a strategy which relies on a concatenation procedure. We will consider languages which have a block structure. On each block, the use of two-way motion will imply savings
of state complexity. The following theorem shows that, for the class $1DFA$, the state complexities of the blocks actually add up. This yields unrestricted freedom in the choice of the block languages.

\begin{theo}[Concatenation Lemma]\label{concat}
Let $L_1,\ldots,L_n$ be non-empty languages over $\Sigma$. Let
$$
L = cL_1cL_2c \ldots cL_nc,
$$
where $c \not\in \Sigma$. Then,
$$
\sigma_0(L)=2+ \sum_{i=1}^n \sigma_0(L_i).
$$
\end{theo}

\begin{proof} Let $M_1,\ldots,M_n$ be the minimal $1DFA$'s with $L(M_i)=L_i$, where

\begin{itemize}
\item $M_i = (Q,\Sigma,\delta_i,q_{0,i},F_i),$
\item $Q_i \cap Q_j = \emptyset$ for $i \not= j$ (wlog).
\end{itemize}
Define $M=(Q,\Sigma \cup \{c\},\delta,q_0,F)$ by (cf.\ Figure~\ref{4})
\begin{eqnarray*}
Q& = &\{q_0,f\} \cup Q_1 \cup \ldots \cup Q_n\\
F& = &\{f\}\\
\delta (p,a)& = & \left\{
  \begin{array}{lll}
   \delta_i(p,a)& {\rm if} & p \in Q_i, \; a \in \Sigma\\
   q_{0,i+1}& {\rm if} & p \in F_i,\; 1 \le i < n, \; a=c\\
   q_{0,1}& {\rm if} & p=q_0,\; a=c\\
   f& {\rm if} & p \in F_n,\; a=c
  \end{array}\right.
\end{eqnarray*}
%
%
%
%
%
%
\begin{figure}
\begin{center}
  \begin{picture}(128,21)
%
%
\put(0,5){\usebox{\sstate}\makebox(0,5){$q_0$}}
\put(120,5){\usebox{\sfstate}\makebox(0,5){$f$}}
%
%
\put(5,8){\vector(1,0){8}}
\put(31,8){\vector(1,0){8}}
\put(58,8){\vector(1,0){8}}
\put(79,8){\vector(1,0){8}}
\put(105,8){\vector(1,0){8}}
%
%
%
\put(8,9){$c$}
\put(34,9){$c$}
\put(61,9){$c$}
\put(82,9){$c$}
\put(108,9){$c$}
\put(70,7){$\cdots$}
%
%
%
\put(14,3){\framebox(10,9){$M_1$}}
\put(41,3){\framebox(10,9){$M_2$}}
\put(88,3){\framebox(10,9){$M_n$}}
\put(24,3){\framebox(5,9){$F_1$}}
\put(51,3){\framebox(5,9){$F_2$}}
\put(98,3){\framebox(5,9){$F_n$}}
  \end{picture}
\end{center}
\caption{Automaton $M$ from the proof of Theorem~\ref{concat}}\label{4}
\end{figure}
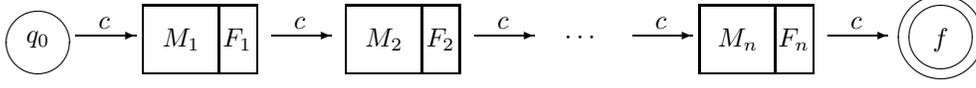
Since $M$ has the desired number of states we are done with the proof if we can show:

\begin{enumerate}
\item $L(M)=L$.
\item $M$ is minimal.
\end{enumerate}

\begin{enumerate}
\item is immediate from the construction.
\item requires that
\begin{itemize}
\item every state in $M$ is reachable.
\item for every $p \not= q$ there exists $w_{p,q}$ with $\delta(p,w_{p,q}) = f \Leftrightarrow \delta(q,w_{p,q}) \not= f$.
\end{itemize}
Reachability follows from $L_i \not= \emptyset$ and reachability in every
$M_i$. Let $p,q \in Q$ with $p \not= q$. We consider the following cases:

\begin{center}
\begin{tabular}{ll}
{\it Case 1:} & one of the two states is $q_0$.\\
& Choose $w_i \in L_i$ and define $w_{p,q} = cw_1c\ldots cw_nc$.\\
{\it Case 2:} & one of the two states is $f$.\\
& Let $w_{p,q} = \varepsilon$.\\
{\it Case 3:} & both states come from $M_i$.\\
& Choose $\tilde{w_{p,q}} \in \Sigma^*$ with $\delta_i(p,\tilde{w_{p,q}})
\in F_i \Leftrightarrow  \delta_i(p,\tilde{w_{p,q}}) \not\in F_i$,\\
& and $w_{i+1} \in L_{i+1}, \ldots , w_n \in L_n$. Define  $w_{p,q}
= \tilde{w_{p,q}}cw_{i+1}c\ldots cw_nc$.\\
{\it Case 4:} & $p \in M_i,\; q \in M_j$, wlog $i<j$.\\
& Choose $\tilde{w_q} \in \Sigma^*$ with $\delta_j(q,\tilde{w_q})
\in F_j$ and $w_{j+1} \in L_{j+1}, \ldots , w_n \in L_n$.\\
& Define $w_{p,q} = \tilde{w_q}cw_{j+1}c\ldots cw_nc$.
\end{tabular}
\end{center}
One checks that the above choices of the $w_{p,q}$ have the desired property.
\end{enumerate}
This concludes the proof.
\end{proof}

Consider the languages
$$
L_i = \{awbwa \; : \; w \in \{0,1\}^i\}
$$
and their concatenation
$$
R_n = cL_1c \ldots cL_nc.
$$

\begin{lemm}\label{1dea}
There exists a $1DFA$ $M_i$ with $4 \cdot 2^i$ states accepting $L_i$. $M_i$ is minimal, that is, $\sigma_0(L_i)=4 \cdot 2^i$.
\end{lemm}

\begin{proof} Define the $1DFA$ $M_i=(Q_i, \{0,1,a,b\}, \delta_i, q_{0,i}, F_i)$ by

\begin{eqnarray*}
Q_i& = &\{(x_1,\ldots ,x_l,{\rm flag}) \; : \; x_j \in \{0,1\}, \;
1 \le l \le i,\\
&&{\rm flag} \in \{{\rm read},\; {\rm compare}\}\} \cup \{q_{0,i},f_i,q_s,q_e\}\\
F_i& = &\{f_i\}\\
\delta(q_{0,i},a)& = &q_s\\
\delta(q_s,x)&=&(x,{\rm read}) \;\;\;\;\; {\rm if} \; x \in \{0,1\}\\
\delta((x_1,\ldots ,x_l,{\rm read}),x)& = & \left\{
\begin{array}{ll}
(x_1,\ldots,x_l,x,{\rm read}) & 1 \le l \le i-1, \; x \in \{0,1\}\\
(x_1,\ldots,x_l,{\rm compare}) & l=i, \; x=b\\
{\rm undefined} & {\rm else}
\end{array}\right.\\
\delta((x_1,\ldots ,x_l,{\rm compare}),x)& = & \left\{
\begin{array}{ll}
(x_2,\ldots,x_l,{\rm compare}) & 2 \le l \le i, \; x=x_1\\
q_e & l=1, \; x=x_1\\
{\rm undefined} & {\rm else}
\end{array}\right.\\
\delta(q_e,a)& = &f_i
\end{eqnarray*}
$M_i$ has the claimed number of states:
$$
|Q_i|=2 \cdot \sum_{l=1}^i 2^l +4 = 2 \cdot (2^{i+1}-2) +4 = 2^{i+2}.
$$
In order to show that $M_i$ is minimal, we have to show reachability and separability of the set of states. Reachability is seen as follows:

\begin{eqnarray*}
q_{0,i} & = & \delta(q_{0,i},\varepsilon)\\
q_s & = &\delta(q_{0,i},a)\\
(x_1,\ldots,x_l,{\rm read}) & = &\delta(q_{0,i},a x_1 \ldots x_l)\\
(x_1,\ldots ,x_l,{\rm compare}) & = &\delta(q_{0,i},a 0^{i-l}x_1 \ldots x_l
b 0^{i-l})\\
q_e & = &\delta(q_{0,i},a 0^i b 0^i)\\
f & = &\delta(q_{0,i},a 0^i b 0^i a)
\end{eqnarray*}
The separating words are given in the following table:
$$
\begin{array}{|c|c|c|}
\hline
{\rm state} & {\rm state} & {\rm separating \; word}\\
\hline
\hline
q_{0,i} & {\rm arbitrary} & a 0^i b 0^i a\\
q_s & {\rm arbitrary} & 0^i b 0^i a\\
q_e & {\rm arbitrary} & a\\
f & {\rm arbitrary} & \varepsilon\\
(x_1, \ldots , x_l , {\rm read}) & {\rm arbitrary} & 0^{i-l} b x_1 \ldots x_l
0^{i-l} a\\
(x_1, \ldots , x_{l_1} , {\rm compare}) & (x_1', \ldots , x_{l_2}' ,
{\rm compare}) & x_1 \ldots x_{l_1} a\\
\hline
\end{array}
$$
Thus, $M_i$ is minimal.
\end{proof}

\begin{lemm}\label{2deak}
There exists a $2DFA$ $N_i$ with $\mathcal{O}(i)$ states accepting $L_i$. We have $k(i) = \lambda (N_i) = \mathcal{O}(i^2)$.
\end{lemm}

\begin{proof} The assertion follows from the proof of Theorem \ref{wcw}.
\end{proof}

\begin{theo}
We have
$$
\sigma_0(R_n) = 2 + 4 \cdot (2^1 + \ldots + 2^n).
$$
\end{theo}

\begin{proof} Follows from Lemma \ref{1dea} and the Concatenation Lemma.
\end{proof}

\begin{theo}\label{hue}
For $1 \le m \le n$ we have
$$
\sigma_{\sum_{i=1}^m k(i)}(R_n) \le 2+c\cdot  ( 1 + \ldots + m + 2^{m+1} +
\ldots + 2^n).
$$
In particular:
$$
\sigma_{\sum_{i=1}^n k(i)}(R_n) \le 2+c \cdot (1 + \ldots + n).
$$
\end{theo}
\begin{proof} The number of allowed left-moves suffices to check the first $m$ blocks in two-way fashion. We therefore can construct the automaton $M \in 2DFA(\sum_{i=1}^m k(i))$ for $R_n$ as given in Figure~\ref{5}.
%
%
%
%
\begin{figure}
\begin{center}
  \begin{picture}(133,25)
%
%
\put(0,5){\usebox{\sstate}\makebox(0,5){$q_0$}}
\put(128,5){\usebox{\sfstate}\makebox(0,5){$f$}}
%
%
\put(5,8){\vector(1,0){7}}
\put(28,8){\vector(1,0){6}}
\put(38,8){\vector(1,0){6}}
\put(60,8){\vector(1,0){7}}
\put(82,8){\vector(1,0){6}}
\put(93,8){\vector(1,0){6}}
\put(115,8){\vector(1,0){7}}
%
%
%
\put(5,9){{\scriptsize $c,+1$}}
\put(28,9){{\scriptsize $c,+1$}}
\put(37,9){{\scriptsize $c,+1$}}
\put(61,9){{\scriptsize $c,+1$}}
\put(83,9){{\scriptsize $c,+1$}}
\put(93,9){{\scriptsize $c,+1$}}
\put(116,9){{\scriptsize $c,+1$}}
\put(35,7){$\cdots$}
\put(88,7){$\cdots$}
%
%
%
\put(12,3){\framebox(10,9){$N_1$}}
\put(45,3){\framebox(10,9){$N_m$}}
\put(67,3){\framebox(10,9){$M_{m+1}$}}
\put(100,3){\framebox(10,9){$M_n$}}
\put(22,3){\framebox(5,9){$F_1$}}
\put(55,3){\framebox(5,9){$F_2$}}
\put(77,3){\framebox(5,9){$F_n$}}
\put(110,3){\framebox(5,9){$F_n$}}
  \end{picture}
\end{center}
\caption{Automation $M$ from the proof of Theorem~\ref{hue}}\label{5}
\end{figure}
This yields the desired estimate for $\sigma_{\sum_{i=1}^m k(i)}(R_n)$.
\end{proof}

Thus, we qualitatively have the picture given by Figure~\ref{6}.
%
%
%
%
%
\begin{figure}
\begin{center}
  \begin{picture}(100,70)
%
%
\put(29,44){$\times$}
\put(39,29){$\times$}
\put(49,22){$\times$}
\put(69,19){$\times$}
%
%
\put(30,10){\vector(1,0){65}}
\put(29,45){\line(1,0){2}}
\put(29,30){\line(1,0){2}}
\put(29,23){\line(1,0){2}}
\put(29,20){\line(1,0){2}}
%
%
%
\put(30,10){\vector(0,1){50}}
\put(30,9){\line(0,1){2}}
\put(40,9){\line(0,1){2}}
\put(50,9){\line(0,1){2}}
\put(70,9){\line(0,1){2}}
%
%
%
\put(94,5){$k$}
\put(15,58){$\sigma_k (R_n)$}
\put(17,45){{\scriptsize $\sigma_0(R_n)$}}
\put(6,30){{\scriptsize $\sigma_{\sum_{i=1}^{m_1} k(i)}(R_n)$}}
\put(6,24){{\scriptsize $\sigma_{\sum_{i=1}^{m_2} k(i)}(R_n)$}}
\put(6,18){{\scriptsize $\sigma_{\sum_{i=1}^n k(i)}(R_n)$}}
\put(29,5){{\scriptsize $0$}}
\put(33,5){{\scriptsize $\sum_{i=1}^{m_1} k(i)$}}
\put(49,5){{\scriptsize $\sum_{i=1}^{m_2} k(i)$}}
\put(67,5){{\scriptsize $\sum_{i=1}^n k(i)$}}
\end{picture}
\end{center}
\caption{Upper bounds established in Theorem~\ref{hue}}\label{6}
\end{figure}
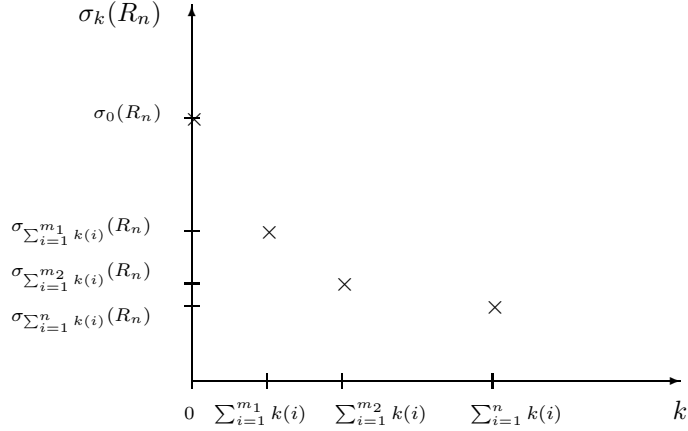

\section{State complexity theorems for the class $2DFA(k)$}

\subsection{Comparison with $1NFA$}

\begin{theo}
For every $k \in \NN$, there exists a constant $c_k$ such that
$$
\begin{array}{ccc} 1NFA & \rightarrow & 2DFA(k)\\
n && \ge \; 2^{c_k \sqrt{n}}
\end{array}
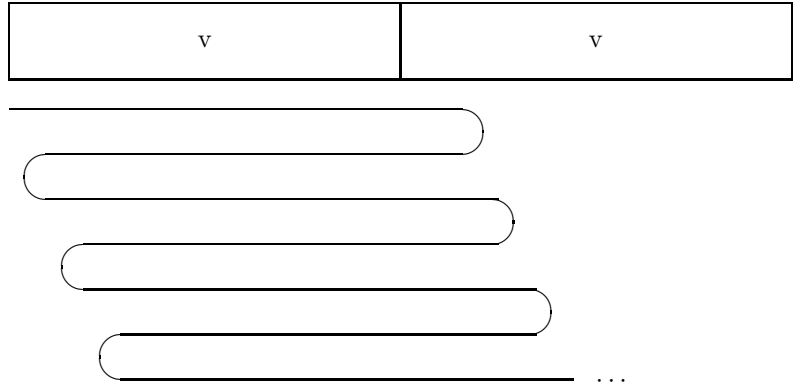
$$
\end{theo}

\begin{proof} Consider the regular languages
$$
L_{n,k}=\{u1^k v \; : \; u,v \in \{0,1\}^n, \; u \not= v\}.
$$
A straightforward modification of the well-known $\mathcal{O} (n^2)$-state $1NFA$ for $\{uv \; : \; u,v \in \{0,1\}^n, \; u \not= v\}$ (cf.\ \cite{w}) yields a $1NFA$ for $L_{n,k}$ with $\mathcal{O} (c'(k) \cdot n^2)$ states. Now, let $M \in 2DFA(k)$ with $L(M)=L_{n,k}$. We have to show that $M$ has at
least $2^n$ states. Consider the words
$$
\{u1^k \; : \; u \in \{0,1\}^n\}.
$$
They are prefixes of words in $L_{n,k}$. Thus, $M$ does not loop on them since $M$ is deterministic. One has the following situation:
%
%
%
%
%
%
\begin{center}
  \begin{picture}(80,40)
%
%
\put(5,30){\framebox(45,10){$u$}}
\put(50,30){\framebox(25,10){$1^k$}}
\put(5,25){$q_0$}
\put(75,1){$q_u$}
%
%
\put(9,26){\line(1,0){16}}
\put(12,20){\line(1,0){13}}
\put(12,14){\line(1,0){22}}
\put(20,8){\line(1,0){14}}
\put(20,2){\vector(1,0){54}}
\put(25,23){\oval(6,6)[r]}
\put(12,17){\oval(6,6)[l]}
\put(34,11){\oval(6,6)[r]}
\put(20,5){\oval(6,6)[l]}
  \end{picture}
\end{center}
The assertion now follows from
$$u_1 \not= u_2 \; \Rightarrow \; q_{u_1} \not= q_{u_2},$$
which again holds by determinism.\end{proof}

\begin{rema}
The result can be improved. Using the languages $ L_n=\{0,1\}^*1\{0,1\}^{n-1} \$ $ we have, by results from Section 3,
$$
\forall n \ge k+1 \;\; : \;\;
\begin{array}{ccc}
1NFA & \rightarrow & 2DFA(k)\\
n+2 && \ge \; 2^n +1
\end{array}
$$
\end{rema}
This provides an analogue of the result of Sipser: By restriction of two-way motion one obtains an exponential trade-off between nondeterminism and two-way motion.

The following theorem gives a result for the other direction.

\begin{theo}
Let $k \in \NN$. Then, there exist constants $c_1,c_2$, such that
$$
\begin{array}{ccc} 2DFA(k) & \rightarrow & 1NFA\\ n && \ge \; \frac{1}{c_1 k} \cdot n \cdot 2^{c_2 \sqrt{k}} \end{array}
$$
\end{theo}

\begin{proof} Consider the languages
$$
L_n=\{wvvw \; : \; wv \in \{0,1\}^n, \; |v| = \lfloor c \cdot \sqrt{k} \rfloor \}
$$
with $c$ to be specified below. In a similar way as in the proof of the corresponding result for
$$
L_n'=\{ww \; : \; w \in \{0,1\}^n\}
$$
as given, for example, in \cite{w}, one checks that every $1NFA$ accepting $L_n$ has at least $2^n$ states. Now, a $2DFA(k)$ can use the following strategy: store the first $n-c\cdot \sqrt{k}$ symbols (i.e. $w$) in the finite control and check the next $2 \cdot c \cdot \sqrt{k}$ symbols by using the $k$ left-moves:
%
%
%
%
%
%
%
\begin{center}
  \begin{picture}(120,70)
%
%
\put(5,50){\framebox(52,10){v}}
\put(57,50){\framebox(52,10){v}}
%
%
\put(5,46){\line(1,0){60}}
\put(10,40){\line(1,0){55}}
\put(10,34){\line(1,0){60}}
\put(15,28){\line(1,0){55}}
\put(15,22){\line(1,0){60}}
\put(20,16){\line(1,0){55}}
\put(20,10){\line(1,0){60}}
\put(65,43){\oval(6,6)[r]}
\put(10,37){\oval(6,6)[l]}
\put(69,31){\oval(6,6)[r]}
\put(15,25){\oval(6,6)[l]}
\put(74,19){\oval(6,6)[r]}
\put(20,13){\oval(6,6)[l]}
\put(83,9){$\cdots$}
  \end{picture}
\end{center}%
Choose $c$ such that the $k$ left moves suffice to test the middle part, compare the proof of Theorem \ref{wcw}. The $2DFA(k)$ is therefore a cartesian product of two automata. We have the following estimate for the number of states:
$$
|M| \le 2^{n-c_2 \sqrt{k}} \cdot c_1 k.
$$
Thus,
$$
|\mbox{min. } 1NFA \mbox{ for } L_n| \; \ge \; \frac{1}{c_1 k} \cdot 2^{c_2 \sqrt{k}} \cdot |\mbox{min. } 2DFA(k) \mbox{ for } L_n|,
$$
concluding the proof.
\end{proof}

\subsection{Comparison with $1DFA$ and $2DFA$}
We are now interested in trade-offs coming from eliminating (resp., allowing) two-way motion, starting from the class $2DFA(k)$, that is,
$$
2DFA \leadsto 2DFA(k) \leadsto 1DFA.
$$

\begin{theo} For every $k \in \NN$, we have
$$
\begin{array}{ccc} 2DFA & \rightarrow & 2DFA(k)\\
5(n+1) + 5  && \ge \; n^n + 1 \end{array}
$$
\end{theo}
This result was already proven in the the previous section.

\begin{theo}
There is a constant $c$ such that for $k \in \NN$, we have
$$
\begin{array}{ccc} 2DFA(k) & \rightarrow & 1DFA\\
n+k+3 && \ge \; c \cdot n \cdot 2^k \end{array}
$$
\end{theo}
The proof of this theorem will be given in two separate lemmas. We
consider the languages
$$
L_{n,k}=\{0,1\}^* 1 \{0,1\}^{k-1}\$_1 \cup \{w \in \{0,1\}^* \; : \;
\#_1(w)=n-1\}\$_2.
$$
First, we give a $2DFA(k)$ for $L_{n,k}$.

\begin{lemm}
For $n,k \in \NN$, there is $M \in 2DFA(k)$ with $L(M)=L_{n,k}$ and
$|M|=n+k+3$
\end{lemm}
\begin{proof} $M$ runs to the rightmost symbol of the input while counting the $1$'s. In the case that the rightmost symbol is $\$_2$, $M$ accepts if and only if  $n-1$ $1$'s were read.
If the rightmost symbol is $\$_1$, $M$ runs back $k$ symbols and checks for a $1$. Formally, let $M=(Q,\{0,1,\$_1,\$_2\},\delta,q_0,F)$ be defined by:

\begin{eqnarray*}
Q & = & \{q_0,\ldots,q_{n+k+1},f\}\\
F & = & \{f\}\\
\delta(q_i,0) & = & (q_i,+1) \; {\rm if} \; 0 \le i \le n\\
\delta(q_i,1) & = & (q_{i+1},+1) \; {\rm if} \; 0 \le i \le n-1\\
\delta(q_n,1) & = & (q_n,+1)\\
\delta(q_i,\$_1) & = & (q_{n+1},-1) \; {\rm if} \; 0 \le i \le n\\
\delta(q_{n-1},\$_2) & = & (f,+1)\\
\delta(q_i,0) & = & (q_{i+1},-1) \; {\rm if} \; n+1 \le i \le n+k-1\\
\delta(q_i,1) & = & (q_{i+1},-1) \; {\rm if} \; n+1 \le i \le n+k-1\\
\delta(q_{n+k},1) & = & (q_{n+k+1},+1)\\
\delta(q_{n+k+1},0) & = & (q_{n+k+1},+1)\\
\delta(q_{n+k+1},1) & = & (q_{n+k+1},+1)\\
\delta(q_{n+k+1},\$_1) & = & (f,+1)
\end{eqnarray*}
$M$ has $n+k+3$ states, belongs to $2DFA(k)$ and accepts $L_{n,k}$.
\end{proof}

Next we want to show that the minimal $1DFA$ for $L_{n,k}$ has at least $c \cdot n \cdot 2^k$ states. To this end, we define a $1DFA$ $M$, show that $M$ accepts $L_{n,k}$, prove minimality
and give an estimate for the number of states of $M$.

\begin{lemm}
For every $M \in 1DFA$ with $L(M)=L_{n,k}$, we have $|M| \ge c \cdot n \cdot
2^k$, where $c \approx \frac{1}{2}$.
\end{lemm}

\begin{proof} Define $M=(Q,\{0,1,\$_1,\$_2\},\delta,q_0,F)$ by

\begin{eqnarray*}
Q & = &  \{(a_1,\ldots,a_k,i) \; : \; a_j \in \{0,1\}, \; i \in
\{0,\ldots,n-1\}, \\
&& \#_1(a_1\ldots a_k) \le i \} \cup \{f\}
\end{eqnarray*}
\begin{eqnarray*}
\delta((a_1,\ldots,a_k,i),a) & = &
 \left\{ \begin{array}{lll}
         (a_2,\ldots,a_k,a,i)&,&a = 0\\
         (a_2,\ldots,a_k,a,i+1)&,&a=1,i<n-1\\
         {\rm undefined}&,&a=1,i \ge n-1
         \end{array}
 \right.\\
\delta((a_1,\ldots,a_k,i),\$_1) & = &
 \left\{ \begin{array}{lll}
         f&,&a_1= 1\\
         {\rm undefined}&&{\rm else}
         \end{array}
 \right.\\
\delta((a_1,\ldots,a_k,i),\$_2) & = &
 \left\{ \begin{array}{lll}
         f&,&i = n-1\\
         {\rm undefined}&&{\rm else}
         \end{array}
 \right.\\
q_0 & = & (0,\ldots,0,0) \in Q\\
F & = & \{f\}.
\end{eqnarray*}
Obviously,
$$L(M)=L_{n,k}.$$
In order to prove minimality of $M$ we first show that every state is reachable:

\begin{eqnarray*}
q_0 & = & \delta(q_0,\varepsilon)\\
f & = & \delta(q_0,1^k\$_1)\\
(a_1, \ldots,a_k,i) & = & \delta(q_0, 1^{i-l}a_1 \ldots a_k)  \mbox{, where } l=\#_1(a_1 \ldots a_k)
\end{eqnarray*}
Now, we have to show separability of states. $(a_1, \ldots,a_k,i)$ and $f$ are separated by $w = \varepsilon$. In order to handle two different $(a_1, \ldots,a_k,i)$ and $(\hat{a_1},
\ldots,\hat{a_k},\hat{i})$ we consider two cases.

{\it Case 1.} $i \not= \hat{i}$ : Wlog, $i > \hat{i}$. Then, the following word separates:
$$
\delta((a_1, \ldots,a_k,i),1^{n-i-1}\$_2)=f,
$$
$$
\delta((\hat{a_1},\ldots,\hat{a_k},\hat{i}),1^{n-i-1}\$_2) \not= f,
$$
since $\$_2$ ensures that the $\hat{a_1},\ldots,\hat{a_k}$ cannot yield acceptance.

{\it Case 2.} $\exists j$ : $a_j=1, \; \hat{a_j}=0$ (wlog)

In this case we have
$$
\delta((a_1, \ldots,a_k,i),0^{j-1}\$_1)=f,
$$
$$
\delta((\hat{a_1},\ldots,\hat{a_k},\hat{i}),0^{j-1}\$_1) \not= f.
$$
Hence, $M$ is minimal and $|M| \approx \frac{1}{2} \cdot n \cdot 2^k$.
\end{proof}

We summarize our results in the following table. The entry $m_{i,j}$ corresponds to the following trade-off:
$$
\begin{array}{ccc} i & \rightarrow & j\\ && \le \\ && \ge \end{array}
$$

{\scriptsize
$$\begin{array}{||l||l|l|l|l||}
\hline\hline
& 2DFA & 2DFA(k) & 1NFA & 1DFA\\
\hline\hline
2DFA && \begin{array}{l} \le n^n \\ \ge 2^{n-4}+1 \end{array}
&\begin{array}{l} \le n^n \\
\ge 2^{c \cdot n} \end{array} &\begin{array}{l} \le n^n \\ \ge
2^{n-3} \end{array} \\
\hline
2DFA(k) & \begin{array}{l} \le n \\ \ge n \end{array} &&\begin{array}{l}
\le n \cdot (k+1) \cdot (|\Sigma|+1)^{k+1} \\
\ge n \cdot \frac{1}{c_1 k}  \cdot c_2^{\sqrt{k}} \end{array} &\begin{array}{l} \le n
\cdot (k+1) \cdot (|\Sigma|+1)^{k+1} \\ \ge c \cdot (n-k-3) \cdot 2^k \end{array} \\
\hline
1NFA & \begin{array}{l} \le 2^n \\ \ge ? \end{array} &\begin{array}{l} \le
2^n \\
\ge 2^{n-2}+1 \end{array} &&\begin{array}{l} \le 2^n \\ \ge
2^n \end{array} \\
\hline
1DFA & \begin{array}{l} \le n \\ \ge n \end{array} &\begin{array}{l} \le
n \\
\ge n \end{array} &\begin{array}{l} \le n \\ \ge n \end{array} &\\
\hline\hline
\end{array}$$
}


\begin{thebibliography}{9999}
\bibitem{b3} J.\ C.\ Birget : {\it Basic techniques for two-way finite automata}, Formal Properties of Finite Automata and Applications, Springer-Verlag, LNCS {\bf 386} (1989), 56--64
\bibitem{b4} J.\ C.\ Birget : {\it Concatenation of inputs in a two-way automaton}, Theor.\ Comput.\ Sci.\ {\bf 63} (1989), 141--156
\bibitem{b5} J.\ C.\ Birget : {\it Two-way automaton computations}, RAIRO, Inform.\ Theor.\ Appl.\ {\bf 24} (1990), 47--66
\bibitem{b1} J.\ C.\ Birget : {\it State-complexity of finite state devices, state compressibility and incompressibility}, Math.\ Syst.\ Theory {\bf 26} (1993), 237--269
\bibitem{b2} J.\ C.\ Birget : {\it Two-way automata and length-preserving homomorphisms}, Math.\ Syst.\ Theory {\bf 29} (1996), 191--226
\bibitem{ds} C.\ Dwork, L.\ Stockmeyer : {\it A time complexity gap for two-way probabilistic finite-state automata}, SIAM J.\ Comput.\ {\bf 19} (1990), 1011--1023
\bibitem{gkw} J.\ Goldstine, C.\ Kintala, D.\ Wotschke : {\it On measuring nondeterminism in regular languages}, Inf.\ Comput.\ {\bf 86} (1990), 179--194
\bibitem{h78} M.\ A.\ Harrison : {\it Introduction to Formal Language Theory}, Addison-Wesley (1978)
\bibitem{h82} K.\ Hashiguchi : {\it Limitedness theorem on finite automata with distance functions}, J.\ Comput.\ System Sci.\ {\bf 24} (1982), 233--244
\bibitem{hu69} J.\ E.\ Hopcroft, J.\ D.\ Ullmann : {\it Formal Languages and their Relation to Automata}, Addison-Wesley (1969)
\bibitem{hu79} J.\ E.\ Hopcroft, J.\ D.\ Ullmann : {\it Introduction to Automata Theory, Languages and Computation}, Addison-Wesley (1979)
\bibitem{mf} A.\ R.\ Meyer, M.\ J.\ Fischer : {\it Economy of description by automata, grammars and formal systems}, 12th Annual SWAT (1971), 188--191
\bibitem{m} S.\ Micali : {\it Two-way deterministic finite automata are exponentially more succinct than sweeping automata}, Inf.\ Process.\ Lett.\ {\bf 12} (1981), 103--105
\bibitem{rs} M.\ Rabin, D.\ Scott : {\it Finite automata and their decision problems}, IBM J.\ Research and Development {\bf 3} (1959), 114--125
\bibitem{ss} W.\ J.\ Sakoda, M.\ Sipser : {\it Nondeterminism and the size of two-way finite automata}, 10th Annual STOC (1978), 275--286
\bibitem{sh} J.\ C.\ Shepherdson : {\it The reduction of two-way automata to one-way automata}, IBM J.\ Research and Development {\bf 3} (1959), 198--200
\bibitem{s} M.\ Sipser : {\it Lower bounds on the size of sweeping automata}, 11th Annual STOC (1979), 360--364
\bibitem{v} M.\ Y.\ Vardi : {\it A note on the reduction of two-way automata to one-way automata}, Inf.\ Process.\ Lett.\ {\bf 30} (1989), 261--264
\bibitem{w} D.\ Wotschke : {\it Beschreibungskomplexit\"at I}, Lecture Notes, Department of Computer Science, Johann Wolfgang Goethe-Universit\"at, Frankfurt am Main, Germany (1994)
\end{thebibliography}
\end{document}